\let\accentvec\vec
\documentclass[]{lipics-v2021}

\let\vec\accentvec
\usepackage[T1]{fontenc}
\usepackage{multirow}
\usepackage{amsmath}
\usepackage{amssymb}
\usepackage{tabu}
\usepackage{booktabs}

\usepackage{thm-restate}
\usepackage{tablefootnote}

\pdfoutput=1
\hideLIPIcs
\nolinenumbers

\ccsdesc[500]{Theory of computation~Computational geometry}

\author{Wouter {Meulemans}}{Department of Mathematics and Computer Science, TU Eindhoven, the Netherlands}{w.meulemans@tue.nl}{https://orcid.org/0000-0002-4978-3400}{Partially supported by the Dutch Research
 Council (NWO) under project number VI.Vidi.223.137.}
 \author{Arjen {Simons}}{Department of Mathematics and Computer Science, TU Eindhoven, the Netherlands}{a.simons1@tue.nl}{https://orcid.org/0009-0008-1271-180X}{Supported by the Dutch Research
 Council (NWO) under project number VI.Vidi.223.137.}
\author{Kevin {Verbeek}}{Department of Mathematics and Computer Science, TU Eindhoven, the Netherlands}{k.a.b.verbeek@tue.nl}{https://orcid.org/0000-0003-3052-4844}{}
 
\authorrunning{W. Meulemans et al.}

\newcommand{\Reals}{\mathbb{R}}
\newcommand{\Fam}{\mathcal{F}}
\newcommand{\spanop}{\mathop{\mathrm{span}}}
\newcommand{\anc}{\mathop{\mathrm{anc}}}

\makeatletter
\renewcommand\subsubsection{\@startsection{subsubsection}{3}{\z@}%
                       {-15\p@ \@plus -4\p@ \@minus -4\p@}%
                       {-0.5em \@plus -0.22em \@minus -0.1em}%
                       {\normalfont\normalsize\bfseries\boldmath}}
\makeatother

\usepackage{graphicx,epstopdf}
\usepackage{wrapfig}
  {\hfill\qed\medskip\par}  

\renewenvironment{proof}[1][Proof]
    {\noindent\textit{#1.  }}  
  {\hfill\qed\medskip\par}  

\begin{document}
\title{Visual Complexity of Point Set Mappings}

%
%

%
\maketitle              
\begin{abstract}
We study the visual complexity of animated transitions between point sets. Although there exist many metrics for point set similarity, these metrics are not adequate for this purpose, as they typically treat each point separately. Instead, we propose to look at translations of entire subsets/groups of points to measure the visual complexity of a transition between two point sets. Specifically, given two labeled point sets $A$ and $B$ in $\Reals^d$, the goal is to compute the cheapest transformation that maps all points in $A$ to their corresponding point in $B$, where the translation of a group of points counts as a single operation in terms of complexity. In this paper we identify several problem dimensions involving group translations that may be relevant to various applications, and study the algorithmic complexity of the resulting problems. Specifically, we consider different restrictions on the groups that can be translated, and different optimization functions. For most of the resulting problem variants we are able to provide polynomial time algorithms, or establish that they are NP-hard. For the remaining open problems we either provide an approximation algorithm or establish the NP-hardness of a restricted version of the problem. Furthermore, our problem classification can easily be extended with additional problem dimensions giving rise to new problem variants that can be studied in future work.

\keywords{Group Translations, Visual Complexity, Point Sets.}
\end{abstract}
\section{Introduction}
\label{sec:introduction}

Visualizations are useful for exploring, analyzing and communicating about complex data~\cite{munzner2014visualization}, leveraging human perceptual and cognitive functions. By visually representing data elements, one may identify patterns, trends and outliers.
Typically, each data element is depicted through a visual representation (symbol, glyph or other visual encoding of its properties). As one of the strongest visual variables~\cite{munzner2014visualization}, the position of this representation is often determined by specific properties, such as (geo)spatial location on a map or other attribute encodings.
Yet, this position may change, for example, as the data updates in time-varying contexts~\cite{archambaultDynamicMaps}, or when a user selects other attributes to map to the element's position. They may also shift to disambiguate visualizations by applying distortions to remove overlap~\cite{EffOverlapRemoval,smwg}. 
In such cases, it is important for users to understand the transition between states~\cite{archambaultDynamicMaps}. They have a mental map~\cite{Tversky-1993} of the old state and need to translate this to the new state. Animated transitions may aid this process \cite{AnimatedTransitions} and are often used to support visual analytics.

We are interested in measuring the \emph{visual complexity} of such transitions, as a proxy to automatically assess the cognitive load they may induce. For simplicity, we assume that each data element is positioned based on a point in Euclidean space (instead of by, e.g., a polygon). The question then, is to measure how complex a transition from one point set to another is.
Various approaches already exist to assess point set similarity, including simple forms such as summing (squared) distances and taking the maximum distance~\cite{smwg}, or more elaborate measures such as the Hausdorff distance~\cite{hausdorff} and Earth-mover's distance~\cite{DBLP:conf/giscience/DuckhamKPSTVW16,earthmovers}.

However, such methods are inadequate for assessing the visual complexity of a transition, as they treat each point ``separately'': how far does it have to move, given some constraints on the entire set. A simple example illustrates that this is not adequate for capturing transition complexity: imagine two point sets such that each point translates by the exact same vector, from one set to the other. Understanding this transition is fairly straightforward, as the user sees but a single, grouped motion, following the Gestalt principle of Common Fate~\cite{gestalt}, as shown in Figure~\ref{fig:CommonFateExample}.
\begin{figure}[h]
    \centering
\includegraphics[scale=0.8]{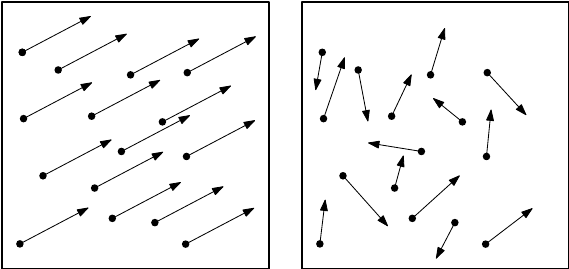}
    \caption{Lower visual complexity when points move in unison.}
    \label{fig:CommonFateExample}
\end{figure}

\subsubsection{Contributions.}
We introduce a problem classification for various problem variants that arise in the study of visual complexity of transitions between point sets, based on common motion of points (Section~\ref{sec:framework}). The problem variants arise from different constraints on the transitions and methods for measuring the complexity. In the subsequent chapters, we study the algorithmic complexity of cases that arise from the classification; refer to Table~\ref{tab:overview} for an overview of our results.

\subsubsection{Related work.}
To improve the understandability of animated transitions between point sets beyond simple linear interpolation, there are roughly two approaches. A first approach is to modify the trajectories along which the points move, to improve the visible structure in the transition by decreasing occlusion and improving grouped motion. Such approaches include trajectory bundling \cite{bundledTrajectories}, the introduction of waypoints \cite{stagedTransitionScatter} and vector-field based routing \cite{vectorFieldTransition}. The drawback is that such techniques tend to increase the complexity of object tracking, as nearby objects can be confused \cite{pointCrowding}. Compared to our work, such techniques take a looser approach to group perception via Common Fate, as they do not build on strict identical movement.

The second approach is to split the animation in several stages -- sequential steps. This can be according to a predefined scheme, for example in the case of waypoints \cite{stagedTransitionScatter}, or optimized explicitly \cite{mizuno2019optimizing}. Mizuno et al. \cite{mizuno2019optimizing} optimize the stages of changing point locations and set memberships, to minimize gaze shift. That is, the goal is to reduce the amount of movement needed for the eyes to track the moving objects through the stages. This is slightly different from our model, as it focuses more on the tracking aspect of transitions \cite{mizuno2019optimizing}, rather than complexity of the mental model that a transition requires.

As mentioned above, our problem relates to point-set similarity. Often, measures are sensitive to applying transformations (such as translations, scaling or rotations) to the objects. When such transformations do not affect a measure, this is called invariance.
In this context, our work is conceptually in between measures that are and measures that are not translation invariant: our models aim to capture common translations, while still accounting for such common translations as affecting similarity.

Aligning point sets via a transformation to optimize their similarity, also known as point-set registration, is a well-studied problem, see e.g. \cite{PointSetRegistrationCor,leastSquares,leastSquares-closed-form} or the survey by Alt and Guibas~\cite{alt2000discrete}. For example, optimizing the sum of squared distances under translations only equates to aligning centroids \cite{cohen1999earth} and is useful, for example, in computing grid maps \cite{eppstein2015improved}; under different measures such alignment may become more complex \cite{eppstein2015improved}. Point-set registrations has applications, for example, also in computer vision to estimate the motion parameters of rigid objects~\cite{leastSquares}. These techniques can be useful to apply a single transformation followed by individual movement, but cannot leverage common patterns that exist only in subsets of the points.
Other methods in this area (e.g. \cite{pointsetRegIterative2,pointsetRegIterative1,coherentPointDrift}) primarily focus on finding the best correspondence between point sets without a labeling, and as such different from the problem studied in this work.

Shape-interpolation techniques also aim to ensure coherent movement of points during transformations. Many of these techniques triangulate the input shapes and interpolate between the resulting triangulations~\cite{as-rigid-as-possible,IntersectFreePolygonMorph,CompatibelTriangulationMorph}. The goal is often to minimize local distortions, ensuring that nearby points move cohesively. 
However, the focus lies on the volume being distorted, rather than precise movements of individual points -- it aims for similar but not necessarily identical transformations.

\section{Problem Classification}
\label{sec:framework}
In this section we present a problem classification for transitions between point sets based on the common motion of points. The problem variants in the classification arise from different constraints and optimization criteria on the transitions.
We discuss potential further extensions/dimensions of the classification in Section~\ref{sec:conclusion}, which then directly introduce new open problems for future work. 

Our input consists of two collections of points $A = \{a_1,\hdots, a_n\}$ and $B = \{b_1,\hdots, b_n\}$ of equal size $n$ in $\Reals^d$. Our goal is to compute the cheapest transformation from points in $A$ to points in $B$, where the translation of any subset of points counts as a single operation in terms of complexity, instead of counting all the individual point translations. 
We refer to such a subset as a \emph{group}\footnote{Not to be confused with the standard notion of \emph{groups} in abstract algebra.}, and to the corresponding translation as a \emph{group translation}. 
A single point may be part of multiple groups. 
We assume that the points in $A$ and $B$ already have a given mapping such that $a_i$ should be translated to $b_i$ for all $i$. In the following, we use the notation $[n] = \{1, \hdots, n\}$. Formally, a solution to this problem consists of a pair $(\Fam, \tau)$, where $\Fam \subseteq \mathcal{P}([n])$ is a family of subsets of $[n]$ indicating the groups via the point indices, and $\tau\colon \Fam\rightarrow \Reals^d$ describes the translation vectors for the groups defined by $\Fam$. We also define the subfamilies $\Fam_i = \{S \in \Fam\mid i \in S\}$ as the set of groups that contain index $i$.
Since the points in $A$ and $B$ have a given mapping, we can represent the input by a single collection of vectors $\Delta = \{\delta_1, \hdots, \delta_n\}$, where $\delta_i = b_i - a_i$ is the difference in position between $a_i$ and $b_i$. Now, a solution $(\Fam, \tau)$ is considered \emph{valid} for an input collection $\Delta$ if $\sum_{S \in \Fam_i} \tau(S) = \delta_i$ for all $i \in [n]$. In the following we will describe two different problem dimensions that can be used to define different variants of the base problem described above.

\subsubsection{Family constraints.}
We consider several restrictions on the family $\Fam$ that can be used in a solution. Such restrictions may be imposed or may be natural in a specific context in which this general optimization problem is applied, and each will lead to a different specific optimization problem.
\begin{description}
    \item[(G)iven:] The family of sets $\Fam$ is specified as part of the input. This is useful when the input data already has an imposed structure that the transformation should follow. To ensure feasibility of the problem, we will assume that $\Fam$ always contains at least all singleton sets.
    \item[(D)isjoint:] All sets in $\Fam$ must be disjoint. This allows for parallel execution of all translations, but strongly restricts the space of valid solutions.
    \item[(H)ierarchical:] The sets in $\Fam$ must form a hierarchy, meaning that for any two sets $S, S' \in \Fam$, either $S \subseteq S'$, $S' \subseteq S$, or $S \cap S' = \emptyset$. This restriction allows for a relatively natural staged transformation, where first the entire collection is translated, and then the translation is broken down into smaller disjoint parts, until all points arrive at their intended destination. Note that, with this restriction, $\Fam$ can efficiently be represented as a rooted tree, where each leaf corresponds to a single point index, and each internal node represents an element of $\Fam$ corresponding to the set of all indices in the subtree of that node. We will also use this tree representation in our results.
    \item[(F)ree:] There is no restriction on $\Fam$, allowing for the best possible transformation. However, it is not directly clear how to effectively visualize the resulting transition. 
\end{description}

\subsubsection{Optimization criteria.}
We also consider two different complexity measures for the transformations represented by $(\Fam, \tau)$. 
\begin{description}
    \item[(C)ardinality:] The number of subsets in $\Fam$, or $|\Fam|$. Note that this complexity measure is not meaningful if $\Fam$ is specified in the input (the ``Given'' variant). 
    \item[(L)ength:] The total length translated by all groups, computed as $\sum_{S \in \Fam} \|\tau(S)\|$, where $\|\cdot\|$ indicates some norm in $\Reals^d$. Unless specified otherwise, we will assume that $\|\cdot\|$ is the Euclidean norm. 
\end{description}
\subsubsection{Naming scheme and organization.}
We study the complexity of the variants of our optimization problem described above. We will mostly focus on the $1$-dimensional and $2$-dimensional versions of the problems, as they are most relevant for the applications we have in mind. To refer to a specific variant of the problem, we use acronyms based on the following naming scheme: \emph{Minimum} [Complexity measure] [Family constraint] \emph{Transformation}, where the complexity measure and family constraint specify one of the options above, using their first letter (in brackets). For example, the MLHT problem refers to the problem where we want to find a solution $(\Fam, \tau)$ for which $\Fam$ represents a hierarchy and where $\sum_{S \in \Fam} \|\tau(S)\|$ is minimized over all (hierarchical) solutions. We also refer to an optimal solution to this problem for some input $\Delta$ simply as an \emph{MLHT} of $\Delta$. Some of our results will require further restrictions on the optimization problem, but we will simply state those restrictions in the respective results.

In Section~\ref{sec:given} we consider variants of the problem where $\Fam$ is specified in the input (the ``Given'' variant). Sections~\ref{sec:length} and $\ref{sec:cardinality}$ cover the remaining family constraints using the different complexity measures: Section~\ref{sec:length} focuses on variants of the problem that use the Length complexity measure, and in Section~\ref{sec:cardinality} we discuss the variants that use the Cardinality complexity measure. We conclude the paper in Section~\ref{sec:conclusion}, where we also describe several additional problem dimensions of the classification that result in new open problems for future work. See Table~\ref{tab:overview} for an overview of our results.


\begin{table}[t]
\caption{Overview of results based on our problem classification.}
\label{tab:overview}
\small
\begin{tabu} to \linewidth {X[3,l,m]|X[1.5,l,m]X[2,l,m]X[2.1,l,m]|X[4.5,l,m]X[1.8,l,m]}
        \toprule
        \bfseries Constraint & 
        \multicolumn{1}{l|}{\bfseries d} &
        \multicolumn{2}{l|}{\bfseries Cardinality} &
        \multicolumn{2}{l}{\bfseries Length}  \\  
        \midrule
        \multicolumn{1}{l|}{\multirow{2}{*}{Given}} & $d=1$ & \multicolumn{1}{|l}{n/a} && $O(n)$ if $\Fam$ is hierarchical & (Th.~\ref{thm:treemedian}) \\ 
        \multicolumn{1}{l|}{} & $d \geq 1$ & \multicolumn{1}{|l}{n/a} & & Polynomial &(Th.~\ref{thm:givenconvex}) \\
        \midrule
        \multicolumn{1}{l|}{Disjoint} & \multicolumn{1}{l|}{$d \geq 1$}  & $O(d n\log n)$ & (Th.~\ref{thm:disjoint}) & $O(d n \log n)$ &(Th.~\ref{thm:disjoint}) \\ 
        \midrule
        \multicolumn{1}{l|}{\multirow{2}{*}{Hierarchical}} & \multicolumn{1}{l|}{$d = 1$} & $O(n\log n)$ & (Th.~\ref{thm:treecount}) & $O(n\log n)$ &(Th.~\ref{thm:MLT}) \\ 
        \multicolumn{1}{c|}{} & \multicolumn{1}{l|}{$d \geq 2$} & $O(d n\log n)$ &(Th.~\ref{thm:treecount}) & NP-Hard & (Th.~\ref{thm:MTLT-NPhard}) \\ 
        \midrule
        \multicolumn{1}{l|}{\multirow{2}{*}{Free}} & \multicolumn{1}{l|}{$d = 1$} & NP-Hard\footnotemark{} &(Th.~\ref{thm:MFCTNPhard}+\ref{thm:MFCTdimreduce}) & $O(n\log n)$ &(Th.~\ref{thm:MLT}) \\ 
        \multicolumn{1}{l|}{} & \multicolumn{1}{l|}{$d \geq 2$} & NP-Hard\addtocounter{footnote}{-1}\footnotemark{} &(Th.~\ref{thm:MFCTNPhard}) & $1.307$-approximation &(Th.~\ref{thm:approx}) \\
        \bottomrule
        \end{tabu}
\end{table}
\addtocounter{footnote}{-1} 
\stepcounter{footnote}\footnotetext{This result is for monotone MCFT (see Section~\ref{sec:cardinality}).}

\begin{theorem}\label{thm:disjoint}
We can compute the MCDT and the MLDT of a collection $\Delta$ of $n$ vectors in $\Reals^d$ in $O(d n \log n)$ time.
\end{theorem}

As for the Disjoint variants of the problem, we state the results directly here, as they are very straightforward. Note that, if $\Fam$ is disjoint, then every index $i$ can be in at most one set $S \in \Fam$, and hence $\Fam_i = \{S\}$. Thus, for a valid solution, we require that $\tau(S) = \delta_i$ for all $i \in S$. As a result, two points can be part of the same group translation if and only if their translation vectors are exactly the same. To optimize for the Cardinality of $\Fam$, we can simply sort $\Delta$ (say, in lexicographical order) and make groups of equal difference vectors. Note that this also directly fixes $\tau$, and hence it also optimizes for Length.

\section{Given Family}\label{sec:given}
In this section we consider the MLGT problem, where we are given both $\Delta$ and $\Fam$ as input. We first state a general result on this problem, before we consider an efficient algorithm for a specific variant in more detail. 

\begin{restatable}{theorem}{MLGTconvex}
\label{thm:givenconvex}
The MLGT problem in $\Reals^d$ is a convex optimization problem, and can hence be solved in (weak) polynomial time.    
\end{restatable}
\begin{proof}
Let $\Fam$ and $\Delta = \{\delta_1, \hdots, \delta_n\}$ be the input of this problem. We can use $\tau$ to encode a solution. A solution is valid if $\sum_{S \in \Fam_i} \tau(S) = \delta_i$ for all $i \in [n]$. Note that these constraints are all linear. Furthermore, the goal is to optimize $\sum_{S \in \Fam} \|\tau(S)\|$. Since all norms on $\Reals^d$ are convex (due to the triangle inequality), the optimization function is also convex. The number of variables in this convex optimization problem corresponds to the number of degrees of freedom in $\tau$, which is $O(d |\Fam|)$. Furthermore, the number of constraints is $O(d n)$. Thus, the resulting convex optimization problem has polynomial size with respect to the input, and hence can be solved in (weak) polynomial time using existing methods such as the ellipsoid method~\cite{ellipsoidMethod}.
\end{proof}

We now consider a variant of the MLGT problem for which we provide an algorithm that runs in linear time. Specifically, we consider the $1$-dimensional version of the problem, and we assume that the given family $\Fam$ is a hierarchy, including all singleton sets. As already explained in Section~\ref{sec:framework} we can hence represent $\Fam$ as a rooted tree $T$ (see Fig.~\ref{fig:CommonFate}). Furthermore, $\Delta$ now simply consists of a collection of real numbers. 

\begin{figure}[h]
    \centering
\includegraphics[width=.4\textwidth]{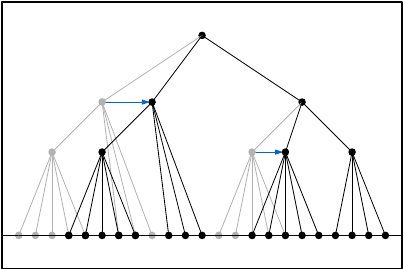}
    \caption{Translating a group in $\Fam$ (blue arrow) corresponds to moving an entire subtree.}
    \label{fig:CommonFate}
\end{figure}

This problem variant is closely related to the median-finding problem, which can be seen as follows. Assume that $\Fam$ consists only of $[n]$ and all singleton sets. If we ignore the cost of the translation of the group $S = [n]$ itself, then the optimal translation for $S$ is exactly a median of all numbers in $\Delta$, because the sum of all individual translations will be minimized. Therefore, this problem can be seen as a generalization of the median-finding problem.

For our algorithm we first define a generalization of a median for a collection of intervals. Let $I =\{[a_1, b_1], \hdots, [a_k, b_k]\}$ be a collection of closed intervals in $\Reals$. Furthermore, define the distance between $x \in \Reals$ and an interval $[a, b]$ as: 
\[
d(x, [a,b]) = \begin{cases}
    0, & \text{if } x \in [a,b] \\
    x - b, & \text{if } x > b \\
    a - x, & \text{if } x < a.
\end{cases}
\]
Now the \emph{interval median} $M(I)$ of a collection of intervals $I$ is the set of real numbers $x$ for which $\sum_i d(x, [a_i, b_i])$ is minimized. 

\begin{restatable}{lemma}{MedianInterval}\label{lem:medianinterval}
The interval median $M(I)$ of a collection of closed intervals $I = \{[a_1, b_1], \hdots [a_k, b_k]\}$ is itself a closed interval. Furthermore, $\sum_i d(x', [a_i, b_i]) \geq \sum_i d(x, [a_i, b_i]) + d(x', M(I))$ for all $x \in M(I)$ and $x' \in \Reals$.    
\end{restatable}
\begin{proof}
For a value $x \in \Reals$, let $L(x)$ be the collection of intervals $[a, b] \in I$ strictly to the left of $x$ ($x > b$). Furthermore, let $L^+(x)$ consist of all intervals $[a, b] \in I$ with $x \geq b$. Similarly, let $R(x)$ be the collection of intervals $[a, b] \in I$ strictly to the right of $x$ ($x < a$), and let $R^+(x)$ consist of all intervals $[a, b] \in I$ with $x \leq a$. Now assume that $|L(x)| > |R^+(x)|$. In that case there exists an $\varepsilon > 0$ such that $d(x-\varepsilon, [a, b]) = d(x, [a, b]) - \varepsilon$ for $[a, b] \in L(x)$, $d(x-\varepsilon, [a, b]) = d(x, [a, b]) + \varepsilon$ for $[a, b] \in R^+(x)$, and $d(x-\varepsilon, [a, b]) = d(x, [a, b])$ for $[a, b] \in I \setminus (L(x) \cup R^+(x))$. Thus, $\sum_i d(x - \varepsilon, [a_i, b_i]) = \sum_i d(x, [a_i, b_i]) + \varepsilon (|R^+(x)| - |L(x)|) < \sum_i d(x, [a_i, b_i])$. Thus, if $x \in M(I)$, then $|L(x)| \leq |R^+(x)|$. By considering the value $x + \varepsilon$ for some sufficiently small $\varepsilon > 0$, we can symmetrically argue that $|R(x)| \leq |L^+(x)|$ for all $x \in M(I)$. Now consider the set of real numbers $x$ such that $|L(x)| \leq |R^+(x)|$. As $|R^+(x)| - |L(x)|$ can only decrease as $x$ increases, and any decrease takes place strictly after the boundary of one of the intervals in $I$, this set of real numbers must take the form of an interval $(-\infty, z]$, where $z$ is a boundary of an interval of $I$. Similarly, the set of real numbers for which $|R(x)| \leq |L^+(x)|$ must take the form of an interval $[y, \infty)$. $M(I)$ can now be obtained as the intersection $[y, z]$ of these intervals, which is thus a closed interval. 

For the second statement of the lemma, consider any $x' \in \Reals$. If $x' \in M(I) = [y, z]$, then the statement holds directly. Instead assume that $x' > z$. Then we have that $|L(x)| > |R^+(x)|$ for all $x \in (z, x']$. Using a similar line of argument as above, together with the fact that $|L(x)| - |R^+(x)| \geq 1$ for all $x \in (z, x']$, this implies that $\sum_i d(x', [a_i, b_i]) \geq \sum_i d(z, [a_i, b_i]) + (x' - z) = \sum_i d(z, [a_i, b_i]) + d(x', M(I))$. The case that $x' < y$ is symmetrical.
\end{proof}

Note that if all intervals consist of a single point, that is, $a_i = b_i$ for all $i$, then $M(I)$ is simply the interval between the lower and upper median of $\{a_1, \hdots, a_k\}$. We will now extend this concept of an interval median to a hierarchy. Let $\Delta$ and $\Fam$ be our input, where $\Delta$ is in $\Reals$ and $\Fam$ is a hierarchy including all singleton groups. We represent $\Fam$ by a rooted tree $T$ where all leaves represent point indices, and all internal nodes represent groups in $\Fam$ (singleton groups correspond to internal nodes with a single leaf child). For an internal node $v \in T$, let $S(v)$ be the corresponding group in $\Fam$. For any node $v \in T$, we can now compute $M(v)$ in a bottom-up fashion as follows: If $v$ is a leaf with index $i$, then $M(v) = [\delta_i, \delta_i]$; otherwise, if $v$ has children $u_1, \hdots, u_k \in T$, then $M(v)$ is the interval median of $\{M(u_1), \hdots, M(u_k)\}$. We can now extend Lemma~\ref{lem:medianinterval} to trees.

\begin{lemma}\label{lem:treemedian}
Let $\Delta$ be a collection of real numbers, and let $\Fam$ be a given hierarchy represented by a rooted tree $T$ with root $r$. There exists a nonnegative constant $c \in \Reals$ such that, for every $x' \in \Reals$, the minimum value of $\sum_{S \in \Fam} |\tau(S)|$ over all valid $\tau\colon \Fam \rightarrow \Reals$ with $\tau(S(r)) = x'$ is at least $d(x', M(r)) + c$.    
\end{lemma}
\begin{proof}
We prove the result by induction on the height of $T$. 

\noindent \textbf{Base case:} $T$ has height $2$, and all children $u_1, \hdots, u_k$ have only a single (leaf) child in $T$. Without loss of generality, we assume that the single child of $u_i$ is the index $i$, for all $i \in [k]$. For $\tau$ to be valid, we require that $\tau(S(u_i)) = \delta_i - x'$, and hence $|\tau(S(u_i))| = d(x', [\delta_i, \delta_i]) = d(x', M(u_i))$. The result now follows directly from Lemma~\ref{lem:medianinterval} and the construction of $M(r)$, where $c$ constitutes as $\sum_i d(x, [a_i, b_i])$, for some $x \in M(I)$.

\noindent \textbf{Induction step:} Let $T$ be a tree of height $h > 2$, and assume that the statement holds for all trees of height smaller than $h$. Again, let $u_1, \hdots, u_k$ be the children of $r$ in $T$. Now consider an optimal valid function $\tau$ with $\tau(S(r)) = x'$. If $x' + \tau(S(u_i)) \notin M(u_i)$, then we can adapt $\tau(S(u_i))$ by the minimum amount $y$ such that $x' + \tau(S(u_i)) \in M(u_i)$. By the induction hypothesis on $u_i$, this will decrease the optimal cost of $\tau$ in the subtree rooted at $u_i$ by at least $y$. Furthermore, the cost of $\tau(S(u_i))$ can increase by at most $y$. Thus, we may assume that there exists an optimal valid function $\tau$ such that $x' + \tau(S(u_i)) \in M(u_i)$ for all $i \in [k]$. Additionally, we can ensure that $|\tau(S(u_i))| = d(x', M(u_i))$, as the optimal cost in the subtree rooted at $u_i$ is not affected as long as $x' + \tau(S(u_i)) \in M(u_i)$, according to the induction hypothesis. The result now follows directly from Lemma~\ref{lem:medianinterval} and the construction of $M(r)$.
\end{proof}

The result in Lemma~\ref{lem:treemedian} can almost directly be used to compute an MLGT for $\Delta$ and $\Fam$. To compute the result efficiently, we use the following lemma.

\begin{restatable}{lemma}{EfficientMedian}\label{lem:efficientmedian}
Let $I = \{[a_1, b_1], \hdots, [a_k, b_k]\}$ be a collection of intervals, then the interval median $M(I)$ can be computed in $O(k)$ time.   
\end{restatable}
\begin{proof}
We will show that $M(I)$ is simply the interval between the lower median and upper median of the collection $\{a_1, b_1, \hdots, a_k, b_k\}$. The interval median can then be computed in $O(k)$ time using standard median-finding algorithms. 

Let $[x, y]$ be the interval between the lower median and upper median of the collection $Z = \{a_1, b_1, \hdots, a_k, b_k\}$. Note that both $x$ and $y$ are endpoints of some (possibly different) intervals. Furthermore, we can choose a partitioning of $Z$ into $Z^-$, $Z^+$, and $\{x, y\}$, where $|Z^-| = |Z^+| = k-1$, $z \leq x$ for all $z \in Z^-$, and $y \leq z$ for all $z \in Z^+$. Now consider any $u \in \Reals$ for which $u < x$. Then at least $k+1$ elements of $Z$, specifically $Z^+ \cup \{x, y\}$, are strictly to the right of $u$. Now assume that there are $r$ intervals in $I$ for which both endpoints are strictly to the right of $u$. Then there must be $k+1-2r$ intervals in $I$ that contain $u$. This leaves $k-1 - (k+1 - 2r) = 2(r-1)$ endpoints in $Z^-$ to make intervals to the left of $u$, of which there can be at most $r-1$. Thus, for some $\varepsilon > 0$, $\sum_i d(u, [a_i, b_i]) > \sum_i d(u+\varepsilon, [a_i, b_i])$, and hence $u \notin M(I)$. We can make a symmetric argument for any $u > y$. Finally, consider a real value $u \in (x, y)$ (assuming that $x \neq y$). Using the same line of argument as above, we can establish that the number of intervals in $I$ strictly to the left of $u$ is equal to the number of intervals in $I$ strictly to the right of $u$. As a result, $\sum_i d(u, [a_i, b_i])$ is equal for all $u \in [x, y]$. Note that we can include the endpoints $x$ and $y$, since the resulting interval must be closed. Thus, $M(I) = [x, y]$. This also holds if $x = y$, since the interval median must exist. 
\end{proof}

The algorithm now simply computes all median intervals in a bottom-up fashion, and then chooses optimal translations inside the median intervals in a top-down fashion. We get the following result.

\begin{restatable}{theorem}{treemedian}\label{thm:treemedian}
Let $\Delta$ be a collection of $n$ real values, and let $\Fam$ be a given hierarchy. We can compute the MLGT of $\Delta$ and $\Fam$ in $O(n)$ time.   
\end{restatable}
\begin{proof}
Let $T$ be the tree that represents $\Fam$. We first compute all interval medians in a bottom-up fashion for all nodes in $T$. By Lemma~\ref{lem:efficientmedian} this takes $O(n)$ time. Let $r$ be the root of $T$ and let $M(r) = [a, b]$. For $\tau(S(r))$ we choose $0$ if $a \leq 0 \leq b$, $a$ if $a > 0$, and $b$ if $b < 0$. By Lemma~\ref{lem:treemedian} this is an optimal choice for $\tau(S(r))$. We then propagate down the tree: for each internal node $v \in T$, let $M(v) = [a, b]$, and let $x$ be the total translation of its strict ancestor nodes. We set $\tau(S(v))$ to $0$ if $a \leq x \leq b$, $a - x$ if $a > x$, and $b - x$ if $b < x$. Again by Lemma~\ref{lem:treemedian}, this is an optimal choice for $\tau(S(v))$ with respect to the subtree of $T$ rooted at $v$. This completes the construction in $O(n)$ time.       
\end{proof}

\section{Optimizing Length}\label{sec:length}
In this section we consider the (remaining) variants of the problem that optimize the Length of the transformation, specifically MLHT and MLFT. 

For the proof of Theorem \ref{thm:MLT} we first provide a lower and upper bound for the optimal Solution of the MLFT in 1D in the two following Lemmas. For that we first define the \emph{span} of a collection $X=\{x_1, \hdots, x_n\}$ of real numbers as $\spanop(X) = \max_i x_i - \min_i x_i$.

\begin{lemma}
\label{lem:MFLT1D}
The MLFT of a collection $\Delta$ in $\Reals$ is at least $\spanop(\Delta \cup \{0\})$.
\end{lemma}
\begin{proof}
Let $(\Fam, \tau)$ be the MLFT of $\Delta$. Consider the collection $\tau^+ = \{\tau(S)\mid S \in \Fam \wedge \tau(S) > 0\}$. We must have that $\sum_{x \in \tau^+} x \geq \max_i \delta_i$, for otherwise the solution $(\Fam, \tau)$ cannot be valid. Similarly, we can construct the collection $\tau^- = \{\tau(S)\mid S \in \Fam \wedge \tau(S) < 0\}$. We must have that $\sum_{x \in \tau^-} x \leq \min_i \delta_i$, for otherwise the solution $(\Fam, \tau)$ can again not be valid. Thus, the total length of $(\Fam, \tau)$ is at least $\max_i \delta_i$ and at least $-\min_i \delta_i$. If $\delta_i > 0$ for all $i \in [n]$, then $\spanop(\Delta  \cup \{0\}) = \max_i \delta_i$, and the result holds. Similarly, if $\delta_i < 0$ for all $i \in [n]$, then $\spanop(\Delta  \cup \{0\}) = -\min_i \delta_i$, and the result holds. Finally, since $\tau^-$ and $\tau^+$ are disjoint, the total length of $(\Fam, \tau)$ is at least $\max_i \delta_i - \min_i \delta_i$, which equals $\spanop(\Delta \cup \{0\})$ in all remaining cases.      
\end{proof}
\begin{lemma}\label{lem:MTLT1D}
The MLHT of a collection $\Delta$ in $\Reals$ is at most $\spanop(\Delta \cup \{0\})$.    
\end{lemma}
\begin{proof}
To prove this result, we simply construct a solution $(\Fam, \tau)$ with the required cost. We first split $\Delta$ into $\Delta^- = \{\delta \in \Delta\mid \delta < 0\}$ and $\Delta^+ = \{\delta \in \Delta\mid \delta > 0\}$. We first sort the values of $\Delta^+$ to obtain the list $\delta_{\sigma(1)}, \hdots, \delta_{\sigma(k)}$, where $\sigma$ is some permutation of $[n]$. We then create the index subsets $S^+_i = \{\sigma(i), \hdots, \sigma(k)\}$ for $i \in [k]$, and add them to $\Fam$. We set $\tau(S^+_1) = \delta_{\sigma(1)}$ and $\tau(S^+_i) = \delta_{\sigma(i)} - \delta_{\sigma(i-1)}$ for $i > 1$. This directly creates a valid transformation for all displacements in $\Delta^+$. Furthermore, $\sum_i |\tau(S^+_i)| = \delta_{\sigma(k)} = \max(\Delta \cup \{0\})$. We can now do the same for the displacements in $\Delta^-$, but instead order them in decreasing order. This results in index subsets $S^-_i$ where $\sum_i |\tau(S^-_i)| = -\min(\Delta \cup \{0\})$. The resulting family $\Fam$ is a hierarchy by construction, and the total length of $(\Fam, \tau)$ is $\sum_i |\tau(S^+_i)| + \sum_i |\tau(S^-_i)| = \max(\Delta \cup \{0\}) - \min(\Delta \cup \{0\}) = \spanop(\Delta \cup \{0\})$.
\end{proof}

We can combine these two bounds into an exact algorithm for both MLFT and MLHT in~1D.


\begin{restatable}
{theorem}{MLT}\label{thm:MLT}
We can compute an MLFT and an MLHT of a collection $\Delta$ of $n$ real numbers in $O(n \log n)$ time.   
\end{restatable}
\begin{proof}
As the MLHT problem imposes stricter requirements on the solution than the MLFT problem, the lower bound of Lemma~\ref{lem:MFLT1D} also applies to the MLHT problem. Thus, by Lemma~\ref{lem:MTLT1D}, an MLHT of $\Delta$ has total length $\spanop(\Delta \cup \{0\})$. Following the construction in Lemma~\ref{lem:MTLT1D}, we can compute an MLHT of $\Delta$ by splitting and sorting the set (twice), which can be done in $O(n \log n)$ time. The resulting family of subsets $\Fam$ may have total size $O(n^2)$, but it can be represented by a tree $T$ with only $O(n)$ nodes. As any solution to the MLHT problem is also a solution to the MLFT problem, the same result applies to that problem.     
\end{proof}

In the 2D version of the problem, the input is specified by a collection $\Delta$ of vectors in $\Reals^2$. We first consider the hierarchical version of the problem, namely MLHT. This problem is essentially the same as the \emph{Euclidean Steiner Tree} (EST) problem, in which input consists of a set of points $P$, and the goal is to compute the shortest tree (measured via Euclidean distance) connecting all points in $P$, where it is allowed to use points not in $P$ as internal nodes for the tree. It is well known that the EST problem is NP-hard in 2D~\cite{SteinerTreeNPhard}. We can show that a MLHT of a collection $\Delta$ of vectors in $\Reals^2$ directly corresponds to an EST of the point set $P = \Delta \cup \{(0 ,0)\}$. This gives us the following result.  



\begin{restatable}{theorem}{MLHThard}\label{thm:MTLT-NPhard}
The MLHT problem is NP-hard in $\Reals^2$.
\end{restatable}
\begin{proof}
    We use a reduction from the Euclidean Steiner Tree problem. Let $P$ be a set of $n$ points in $\Reals^2$. Choose any point $p^* \in P$ and let $\Delta$ be obtained by translating $P$ such that $p^*$ is at the origin, that is, $\Delta = \{p - p^*\mid p \in P\setminus\{p^*\}\}$. We now show that there exists an EST $Y$ on $P$ of length at most $L$ if and only if the MLHT of $\Delta$ has length at most $L$. 

\begin{wrapfigure}[13]{r}{0.30\textwidth}
    \centering
    \vspace{-0.5\baselineskip}
\includegraphics[width=0.29\textwidth]{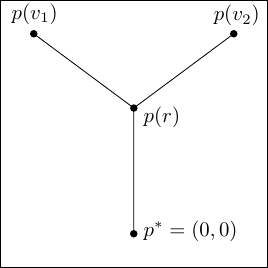}
    \caption{An EST of a tree with three nodes.}
    \label{fig:Steiner}
\end{wrapfigure}

Let $(\Fam, \tau)$ be a an MLHT of $\Delta$ of length at most $L$. Let $T$ be the tree that represents $\Fam$, and let $S(v) \in \Fam$ be the group that corresponds to a node $v \in T$. We assume that the root $r$ of $T$ always corresponds to $S(r) = [n]$, as otherwise we can always add it with $\tau(S(r)) = (0, 0)$ without changing the cost. We now construct a Steiner tree $Y$ from $(\Fam, \tau)$ as follows. For every internal node $v \in T$ we introduce a point $p(v) \in \Reals^2$. We set $p(r) = \tau(S(r))$. For all other internal nodes $v \in T$, we set $p(v) = p(u) + \tau(S(v))$ in a top-down fashion, where $u$ is the parent of $v$ in $T$. The Steiner Tree $Y$ now consists of the segment from $(0, 0)$ to $p(r)$, plus all segments from $p(u)$ to $p(v)$, where $v \neq r$ is an internal node in $T$ and $u$ is its parent (see Fig.~\ref{fig:Steiner}). By construction, $Y$ is clearly connected and includes the point $(0, 0)$. Now consider a vector $\delta_i \in \Delta$. Since $(\Fam, \tau)$ is a valid solution, we get that $\sum_{S \in \Fam_i} \tau(S) = \delta_i$. Let $S_1, \hdots, S_k$ be the sets in $\Fam_i$ ordered decreasingly on cardinality. Since $\Fam$ is hierarchical, we must have that $S_{j+1} \subset S_j$ for all $1 \leq j < k$. Furthermore, if $S_{j+1}$ corresponds to a node $v \in T$, then $i \in S(u)$ for the parent $u$ of $v$, and hence $S(u) = S_j$. Finally, by assumption we have that $S_1 = S(r)$. Thus, if $v \in T$ is the parent of a leaf representing a singleton group $\{i\}$, then we get that $p(v) = \sum_{u \in \anc(v)} \tau(S(u)) = \delta_i$, where $\anc(v)$ is the set of ancestors of $v$ in $T$. Hence $Y$ is a valid Steiner Tree for $\Delta \cup \{(0,0)\}$. Furthermore, it is easy to see that the total length of $(\Fam, \tau)$ and $Y$ is exactly the same. By shifting $Y$ by the vector $p^*$, we obtain an EST on $P$ with length at most $L$.

Now assume that $Y$ is an EST on $P$ of length at most $L$, and let $Y'$ be obtained by shifting $Y$ by $-p^*$, such that $Y'$ is a Steiner Tree on $\Delta \cup \{(0, 0)\}$. We may assume that $Y'$ is a tree, for otherwise we can remove parts of $Y'$ to make it a tree whilst keeping its validity and not increasing the length of $Y'$. We now construct a solution $(\Fam, \tau)$ as follows. We first root this tree in the point $(0, 0)$. For all points $p$ that are either Steiner points of $Y'$ or vectors in $\Delta$, add a set $S$ to $\Fam$ corresponding to the points in $\Delta$ that are in the subtree of $p$ in $Y'$ (if $S = \emptyset$, then we can ignore this Steiner point, as it is not needed). Set $\tau(S)$ to the difference between $p$ and its parent in $Y'$ (Steiner point of $Y'$ or vector in $\Delta \cup \{(0, 0)\}$). It is now easy to see that $(\Fam, \tau)$ is indeed a valid solution for the MLHT problem on $\Delta$. As every segment of $Y'$ is represented by at most one $\tau(S)$ for some $S \in \Fam$, the solution $(\Fam, \tau)$ also has length at most $L$.  
\end{proof}

We next turn our attention to the Free variant of the problem. For this problem variant we do not know whether the problem is NP-hard, or if it can be solved in polynomial time. Instead, we show a very simple and efficient approximation algorithm for the problem. 

We first consider a version of the MLFT problem where, instead of using the Euclidean distance to measure the length of a single translation, we use the Manhattan distance. That is, the length of a translation by $(x, y)$ is $|x| + |y|$. We refer to this problem using the Manhattan distance as the \emph{Manhattan MLFT} problem. 
Now let $(\Fam, \tau)$ be the optimal solution to the Manhattan MLFT problem for some collection of vectors $\Delta$ in $\Reals^2$. Let $S \in \Fam$ with $\tau(S) = (x, y)$. We can replace $S$ by two copies $S'$ and $S''$ of $S$, and set $\tau(S') = (x, 0)$ and $\tau(S'') = (0, y)$. Since the length of a translation is measured by the Manhattan distance, the resulting solution has exactly the same cost, and is hence still optimal. This implies that the Manhattan MLFT problem can be optimally solved by first solving the problem optimally on $x$-coordinates, and then solving the problem optimally on $y$-coordinates independently. By applying Theorem~\ref{thm:MLT} twice, we obtain the following result.

\begin{lemma}\label{lem:Manhattan}
The Manhattan MLFT of a collection $\Delta$ of $n$ vectors in $\Reals^2$ can be computed in $O(n \log n)$ time.   
\end{lemma}

Observe that the Manhattan MLFT is also a valid solution for the (Euclidean) MLFT problem. In fact, it can be shown that the Manhattan MLFT is a $\sqrt{2}$-approximation. Instead of proving this result, we present a slight improvement below. Note that we could also use the Manhattan distance for some rotated set of axes. Specifically, let the $\beta$-Manhattan MLFT for some collection of vectors $\Delta$ be the transformation obtained by first rotating $\Delta$ clockwise around the origin by angle $\beta$, then computing the Manhattan MLFT $(\Fam, \tau)$ on the resulting vectors $\Delta'$, and then rotating all individual translations in $\tau$ counterclockwise by angle $\beta$. Observe that the resulting solution $(\Fam, \tau')$ is again a valid solution for the (Euclidean) MLFT problem.

\begin{restatable}{theorem}{MLFTapprox}
\label{thm:approx}
Either the ($0$-)Manhattan MLFT or the $\frac{\pi}{4}$-Manhattan MLFT of some collection $\Delta$ of vectors in $\Reals^2$ is a $c$-approximation of the (Euclidean) MLFT of $\Delta$ with $c = \sin(\frac{\pi}{8}) + \cos(\frac{\pi}{8}) \approx 1.307$.  
\end{restatable}
\begin{proof}
Let $(\Fam, \tau)$ be the (Euclidean) MLFT of $\Delta$ and let $L$ be its total length. Let $S_1, \hdots, S_k$ be the groups in $\Fam$. For some group $S_i$ let $\alpha_i$ ($0 \leq \alpha_i \leq \frac{\pi}{4}$) be the smallest angle between the vector $\tau(S_i)$ and either the $x$- or $y$-axis. If we denote the Euclidean norm using $\|\cdot\|$ and the Manhattan norm using $\|\cdot\|_1$, then, using basic trigonometry, we obtain that $\|\tau(S_i)\|_1 = (\sin(\alpha_i) + \cos(\alpha_i)) \|\tau(S_i)\|$. Now let $\alpha'_i$ ($0 \leq \alpha'_i \leq \frac{\pi}{4}$) be the smallest angle between the vector $\tau(S_i)$ and either of the rotated axes of the $\frac{\pi}{4}$-Manhattan MLFT. Observe that $\alpha_i + \alpha'_i = \frac{\pi}{4}$. Consider the following sum:
\begin{align*}
\sum_{i=1}^k \alpha_i \|\tau(S_i)\| + \sum_{i=1}^k \alpha'_i \|\tau(S_i)\| &= \\
\sum_{i=1}^k (\alpha_i + \alpha'_i) \|\tau(S_i)\| &= \frac{\pi}{4} L
\end{align*}
Thus either $\sum_{i=1}^k \alpha_i \|\tau(S_i)\|$ or $\sum_{i=1}^k \alpha'_i \|\tau(S_i)\|$ is at most $\frac{\pi}{8} L$. Without loss of generality, assume that this holds for the first sum. We now consider the sum $\sum_{i=1}^k (\sin(\alpha_i) + \cos(\alpha_i)) \|\tau(S_i)\|$, which corresponds to the length of $(\Fam, \tau)$ in Manhattan distance. Note that the function $\sin(\alpha) + \cos(\alpha)$ is strictly concave for $0 \leq \alpha \leq \frac{\pi}{4}$. We claim that, in the worst case, all angles $\alpha_i$ should be equal, assuming that the sum $\sum_{i=1}^k \alpha_i \|\tau(S_i)\|$ is fixed. For the sake of contradiction, assume that $\alpha_i < \alpha_j$ in the worst case for some $i$ and $j$. Then we can replace both $\alpha_i$ and $\alpha_j$ by $\alpha = \left(\frac{\|\tau(S_i)\|}{\|\tau(S_i)\| + \|\tau(S_j)\|}\right) \alpha_i + \left(\frac{\|\tau(S_j)\|}{\|\tau(S_i)\| + \|\tau(S_j)\|}\right) \alpha_j$. It is easy to verify that the sum $\sum_{i=1}^k \alpha_i \|\tau(S_i)\|$ remains the same and that $\alpha_i < \alpha < \alpha_j$. However, since the function $\sin(\alpha) + \cos(\alpha)$ is strictly concave, the sum $\sum_{i=1}^k (\sin(\alpha_i) + \cos(\alpha_i)) \|\tau(S_i)\|$ will increase, which contradicts that $\alpha_i < \alpha_j$ in the worst case.

Thus, we obtain that $\alpha_i = \frac{\pi}{8}$ for all $1 \leq i \leq k$ in the worst case. Hence, the length of $(\Fam, \tau)$ in Manhattan distance is at most $\sum_{i=1}^k (\sin(\frac{\pi}{8}) + \cos(\frac{\pi}{8})) \|\tau(S_i)\| = (\sin(\frac{\pi}{8}) + \cos(\frac{\pi}{8})) L$. Since the Manhattan MLFT is at least as good as $(\Fam, \tau)$, we can conclude that the Manhattan MLFT is a $c$-approximation of the (Euclidean) MLFT with $c = \sin(\frac{\pi}{8}) + \cos(\frac{\pi}{8})$.
\end{proof}
\begin{figure}[h]
    \centering
\includegraphics{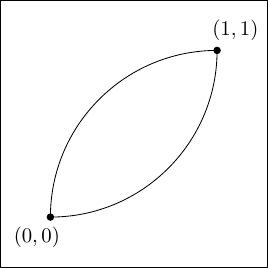}
    \caption{Two quarter arcs of the unit circle with a rotational symmetry around the point~(0.5,~0.5).}
    \label{fig:half_arcs}
\end{figure}
Based on Theorem~\ref{thm:approx}, the following natural question arises: can the approximation factor be improved by considering the $\beta$-Manhattan MLFT for even more angles? Although this may indeed be the case, we sketch that there is a limit to this approach using a simple example. Consider the collection of vectors/points $\Delta$ obtained by densely sampling points along two quarter arcs of the unit circle, as depicted in Figure~\ref{fig:half_arcs}. For this example it is important that the points on both arcs are sampled such that there is a rotational symmetry around the point $(0.5, 0.5)$. It is easy to verify that, regardless of the angle $\beta$, the length of the $\beta$-Manhattan MLFT approaches $2$ when the arcs are sampled densely enough. However, we can construct a (Euclidean) MLFT $(\Fam, \tau)$ by choosing appropriate groups $S_i$ along with translations $\tau(S_i)$ corresponding to the difference vectors between all consecutive points along the bottom arc. This allows us to also construct the vectors on the top arc, if we choose the groups $S_i$ carefully. The total length of $(\Fam, \tau)$ approaches $\frac{\pi}{2}$ when the arcs are sampled densely enough. Thus, even if we consider the $\beta$-Manhattan MLFT for all angles $\beta$, we can never achieve an approximation ratio better than $\frac{4}{\pi} \approx 1.273$.

Finally, we want to observe that the Manhattan distance does not always make our base problem easier. If we consider the hierarchical version of the problem, then we can simply follow the same reduction as in Theorem~\ref{thm:MTLT-NPhard}, but instead reducing from the Rectilinear Steiner Tree Problem. Since this problem is also NP-hard~\cite{RectSteinerTreeNPhard}, it is also NP-hard to compute the Manhattan MLHT in~$\Reals^2$.

\section{Optimizing Cardinality} \label{sec:cardinality}

In this section we consider the variants of the problem that minimize the number of groups in $\Fam$. Before we discuss the more challenging Free variants, we first briefly discuss the hierarchical variant. 

Let $(\Fam, \tau)$ be the MCHT of some collection of vectors $\Delta$ in $\Reals^d$. Assume that there are two groups $S$ and $S'$ in $\Fam$ such that $S \subset S'$, and there is no $S'' \in \Fam$ such that $S \subset S'' \subset S'$. We can then replace $\tau(S)$ by $\tau(S) + \tau(S')$ and replace $S'$ by $S'\setminus S$. If $S'\setminus S$ was already in $\Fam$, then we can simply add $\tau(S')$ to $\tau(S'\setminus S)$ (if $S'\setminus S = \emptyset$, then the group can simply be removed). Note that this operation does not invalidate the solution and also does not increase the number of groups in $\Fam$. Furthermore, it breaks the subset relation between $S$ and $S'$ and does not introduce any new subset relations between groups. Thus, we can repeat this operation until all groups are disjoint. Finally observe that a disjoint family of groups is also hierarchical. The following result now follows from Theorem~\ref{thm:disjoint}.

\begin{theorem}\label{thm:treecount}
We can compute the MCHT for a collection $\Delta$ of $n$ vectors in $\Reals^d$ in $O(d n \log n)$ time.
\end{theorem}

We will now consider the Free variant (MCFT) of this problem. Although we are unable to establish the algorithmic complexity of this problem, we will show that a restricted version of the problem, the \emph{monotone} MCFT, is NP-hard. In the monotone version of the problem, all components of vectors in $\Delta$ must be non-negative, and the same requirement holds for all vectors in $\tau$ in a solution $(\Fam, \tau)$.
In the following we first allow that the dimension of $\Delta$ is specified as part of the input. Then we show NP-hardness via a reduction from Vertex Cover~\cite{Karp1972}. In the Vertex Cover problem we are given a graph $G = (V, E)$ and a number $k$, and the goal is to determine if there exists a subset $C \subseteq V$ with $|C| \leq k$ such that, for every edge $(u, v) \in E$, either $u \in C$ or $v \in C$. Our reduction is very similar to the NP-hardness proof of the Normal Set Basis problem by Jiang and Ravikumar~\cite{minimalNFAproblems}, but requires some extra work.   

\begin{restatable}{theorem}{MFCThard}\label{thm:MFCTNPhard}
The monotone MCFT problem in $\Reals^d$ is NP-hard if $d$ is given as part of the input.
\end{restatable}
\begin{proof}
We use a reduction from Vertex Cover to the decision version of the monotone MCFT problem. Let $G = (V, E)$ be the input graph of a Vertex Cover instance. We will now construct an instance $\Delta$ for the monotone MCFT problem as follows. For every vertex $v \in V$ we consider 2 vectors $x_1(v)$ and $x_2(v)$. These vectors must be distinct (also for different vertices) and each have one component set to $1$ and all other components set to $0$ (like a basis vector). Similarly, for every edge $e \in E$ we consider 4 vectors $x_1(e), \ldots, x_4(e)$, again with one component set to $1$ and all other components set to $0$. As every vector must have a unique component set to $1$, all vectors have dimension $d = 2 |V| + 4 |E|$. Now, for every vertex $v \in V$ we add $\delta(v) = x_1(v) + x_2(v)$ to $\Delta$. Furthermore, for every edge $e = (u, v) \in E$ we add the following vectors to $\Delta$:
\begin{equation*}
    \begin{split}
        \delta_1(e) &= x_1(u) + x_1(e) + x_2(e), \\
        \delta_2(e) &= x_2(v) + x_2(e) + x_3(e), \\
        \delta_3(e) &= x_2(u) + x_3(e) + x_4(e), \\
        \delta_4(e) &= x_1(v) + x_4(e) + x_1(e), \\
        \delta_5(e) &= x_1(e) + x_2(e) + x_3(e) + x_4(e).
    \end{split}
\end{equation*}
We now show that $G$ has a vertex cover of size at most $k$ if and only if $\Delta$ has a monotone MCFT $(\Fam, \tau)$ with $|\Fam| \leq n + 4 |E| + k$.

First assume that $G$ has a vertex cover $C$ of size $k$. We will now first construct (the image of) $\tau$ before we assign the actual groups to the translations in $\tau$. For every $v \in C$ we add both $x_1(v)$ and $x_2(v)$ in $\tau$. For every $v \in V \setminus C$ we add $x_1(v) + x_2(v)$ to $\tau$. For every edge $e = (u, v) \in E$, assume without loss of generality, that $u \in C$. Then we add the vectors $\delta_2(e)$, $\delta_4(e)$, $x_1(e) + x_2(e)$, and $x_3(e) + x_4(e)$ to $\tau$. Observe that we have added exactly $n + 4 |E| + k$ translations to $\tau$. Furthermore, we can construct every vector $\delta(v) \in \Delta$ and $\delta_1(e), \hdots, \delta_5(e) \in \Delta$ by summing a subset of translations in $\tau$. For $\delta(v)$ this is trivial, as $\tau$ either contains $\delta(v)$ or both $x_1(v)$ and $x_2(v)$. For $\delta_2(e)$ and $\delta_4(e)$ this is also trivial, since both are in $\tau$. Furthermore, $\delta_1(e) = x_1(u) + (x_1(e) + x_2(e))$, and $x_1(u) \in \tau$ since $u \in C$ by assumption. Similarly, $\delta_3(e) = x_2(u) + (x_3(e) + x_4(e))$. Finally, we also have that $\delta_5(e) = (x_1(e) + x_2(e)) + (x_3(e) + x_4(e))$. It is now straightforward to fill the groups in $\Fam$ with the right elements such that $(\Fam, \tau)$ is a valid solution for $\Delta$ with $|\Fam| = n + 4 |E| + k$.

Now assume that there exists a monotone MCFT $(\Fam, \tau)$ of $\Delta$ with $|\Fam| \leq n + 4 |E| + k$. We will assign the groups in $\Fam$ to vertices or edges of $G$. We can assign a group $S \in \Fam$ to a vertex $v \in V$ if all components of $\tau(S)$ are at most that of $\delta(v)$. Note that any other group cannot contribute to $\delta(v)$, since we require the transformation to be monotone. Similarly, we can assign a group $S \in \Fam$ to an edge $e \in E$ if all components of $\tau(S)$ are at most that of $\delta_1(e)$, $\delta_2(e)$, $\delta_3(e)$, $\delta_4(e)$, or $\delta_5(e)$, and it is not already assigned to a vertex. At least one group in $\Fam$ must be assigned to each vertex $v \in V$, for otherwise we cannot construct $\delta(v)$. Furthermore, at least $4$ groups must be assigned to each edge $e \in E$. For the sake of contradiction, assume that only $3$ groups are assigned to an edge $e$. Then there must be a group $S$ involved in constructing two vectors out of $\delta_1(e), \delta_2(e), \delta_3(e)$, and $\delta_4(e)$. If these two vectors are $\delta_1(e)$ and $\delta_3(e)$ (or $\delta_2(e)$ and $\delta_4(e)$), then $\tau(S)$ is the $0$-vector due to monotonicity, which is a contradiction. Otherwise assume that these vectors are $\delta_1(e)$ and $\delta_2(e)$ (other cases are symmetric). Then $\tau(S)$ may only have a non-zero component where $x_2(e)$ has a $1$. As a result, we already need two other groups to form the vectors $\delta_1(e)$ and $\delta_2(e)$, and neither group translation can have a non-zero component where $x_4(e)$ has a $1$. Thus, we cannot construct $\delta_3(e)$ (or $\delta_4(e)$), which is a contradiction.

Finally, assume that for edge $e=(u,v)\in E$ neither $u$ nor $v$ has been assigned more than one group. This means that the groups assigned to $u$ and $v$ cannot be used to construct the vectors $\delta_1(e), \hdots, \delta_5(e)$. Since these vectors are independent, we need at least $5$ groups assigned to $e$ in this case. If an edge has at least $5$ groups assigned to it, then we can always reassign one of the groups assigned to $e$ to either $u$ or $v$ instead. Thus, we obtain that every edge has exactly $4$ groups assigned to it, every vertex has at least $1$ group assigned to it, and for every edge there is at least one incident vertex to which we have assigned at least $2$ groups. To construct the vertex cover $C$, we simply take all vertices in $V$ to which we have assigned at least $2$ groups in $\Fam$. By construction $C$ satisfies the constraints of a vertex cover, and it contains at most $|\Fam| - 4 |E| - n = k$ elements, as required. 

The construction can trivially be computed in polynomial time, and hence it is NP-hard to compute the monotone MCFT of $\Delta$, if the dimension $d$ of $\Delta$ is part of the input.
\end{proof}

Theorem~\ref{thm:MFCTNPhard} does not apply to cases where the number of dimensions $d$ of $\Delta$ is small. However, we show that the (monotone) $d$-dimensional version of the problem, even when $d$ is polynomial in $n$, can be reduced to the $1$-dimensional version of the problem under some mild assumptions.

\begin{restatable}{theorem}{MFCTdim}\label{thm:MFCTdimreduce}
There exist a polynomial time reduction of the (monotone) MCFT problem in $d$ dimensions to the (monotone) MCFT problem in $1$ dimension, assuming that $d$ is polynomial in $n$ and the input $\Delta$ consists of bounded integer vectors in $\Reals^d$. 
\end{restatable}
\begin{proof}
We will first focus on the non-monotone version of the problem.
Let $\Delta$ be the input for the MCFT problem in $d$ dimensions, and let $M$ be the maximum absolute value of integers in $\Delta$. We construct a linear function $f\colon \Reals^d \rightarrow \Reals$ that maps vectors in $\Delta$ to numbers in $\Delta'$. Let $\delta_i \in \Delta$ be denoted as $(\delta_i^1, \hdots, \delta_i^d)$, then $f$ is defined as follows:
\[
\delta'_i = f(\delta_i) = \sum_{j=1}^d \delta_i^j A^{j-1},
\]
where $A = 3(n+1)!^2 M$. We claim that the MCFT of $\Delta$ (in $d$ dimensions) has equal cost to the MCFT of $\Delta'$ (in $1$ dimension). 

Let $(\Fam, \tau)$ be an MCFT of $\Delta$. Then, by definition, we have that $\sum_{S \in \Fam_i} \tau(S) = \delta_i$. We claim that $(\Fam, f(\tau))$ is a valid solution for $\Delta'$, where $f(\tau)$ is defined as $f(\tau)(S) = f(\tau(S))$ for all $S \in \Fam$. By linearity of $f$ we get that:
\begin{align*}
\sum_{S \in \Fam_i} f(\tau)(S) &= \sum_{S \in \Fam_i} f(\tau(S))\\
&= f\left(\sum_{S \in \Fam_i} \tau(S)\right)\\
&= f(\delta_i) = \delta'_i.
\end{align*}
Thus, $(\Fam, F(\tau))$ is indeed a valid solution for $\Delta'$.

Now let $(\Fam, \tau')$ be an MCFT of $\Delta'$. Our aim is to construct an inverse mapping of $f$ to map numbers in $\tau'$ to $d$-dimensional vectors. To that end we define a number $x$ as \emph{proper} if it can be written as $x = \frac{1}{D} \sum_{j=1}^d x^j A^{j-1}$, where $D$ is an integer with $|D| \leq n!$ and all $x^j$ are integers with $|x^j| \leq n!(n+1)! M$ for all $1 \leq j \leq d$. If some integers $x^j$ are larger, but still satisfy $|x^j| \leq (n+1)!^2 M$, then we call $x$ \emph{well-behaved}.
If $x$ is well-behaved, then we can compute a corresponding integer vector $y$ in $d$ dimensions as follows. We first compute $x' = D x$, which is an integer. Then, for the first component, we first compute $y^1 = (x' \mod A)$. If $y^1 \leq (n+1)!^2 M$, then we replace $x'$ by $(x' - y^1)/A$ and repeat the process for the remaining components. Otherwise, if $y^1 \geq 2(n+1)!^2 M$, then we set $y^1$ to $y^1-A$ (it becomes negative) and repeat the process with $x'$ replaced by $(x'+y^1)/A$. Observe that the case $(n+1)!^2 M < y^1 < 2(n+1)!^2 M$ cannot occur if $x$ is well-behaved. Furthermore, we get that $y^j = x^j$ by construction. Finally, we divide every component of $y$ by $D$. 
We refer to the result $y$ of this process on a well-behaved number $x$ simply as $f^{-1}(x)$. It is easy to verify that, for all well-behaved numbers $x$, $f$ is the inverse of $f^{-1}$, that is, $f(f^{-1}(x)) = x$. That also directly implies that $f^{-1}$ is a linear function on well-behaved numbers. 

Now assume that for all $S \in \Fam$ we have that $\tau'(S)$ is proper. It follows directly from the definitions that any (partial) sum of at most $n$ proper numbers is well-behaved. We claim that $(\Fam, f^{-1}(\tau'))$ is a valid solution for $\Delta$, where $f^{-1}(\tau')(S) = f^{-1}(\tau'(S))$ for all $S \in \Fam$. Since $f^{-1}$ is linear for all well-behaved numbers, we can use the same argument as above to show that $\sum_{S \in \Fam_i} f^{-1}(\tau')(S) = \delta_i$ for all $i$. Here we use that $\tau'(S)$ is proper for all $S \in \Fam$ and that $|\Fam_i| \leq |\Fam| \leq n$.

Finally, we show that, if $(\Fam, \tau')$ is optimal, then indeed $\tau'(S)$ must be proper for all $S \in \Fam$. We can see the numbers $\tau'(S_i)$ for all $S_i \in \Fam$ as a column vector $\tau'$ with $k = |\Fam|$ dimensions. Similarly we can see $\Delta'$ as a column vector with $n$ dimensions. Since $(\Fam, \tau')$ is a valid solution for $\Delta'$, there must exist a binary $n \times k$ matrix $B$ such that $B \tau = \Delta'$. We can then remove rows from $B$ (and $\Delta'$) until all rows of $B$ are linearly independent. We refer to the reduced form of $\Delta'$ as $\Delta''$. Now assume that $B$ is not full rank. In that case we can add a row to $B$ that corresponds to an equation of the form $\tau'(S_i) = 0$. Since the resulting set of linear equations must have a solution, this implies that there exists a valid solution $(\Fam\setminus\{S_i\}, \tau'')$ for $\Delta'$, which is a contradiction. Therefore, $B$ must have full rank. In that case we obtain that $\tau' = B^{-1} \Delta''$. As the rank of $B$ is at most $n$, it follows from Cramer's rule that all elements in $B^{-1}$ are of the form $\frac{a}{\det(B)}$, where $\det(B)$ is the determinant of $B$, and $a$ is the determinant of a submatrix of $B$. Thus, $a$ and $\det(B)$ are both integers with $|a|, |\det(B)| \leq n!$. Recall that every $\delta'_i \in \Delta'$ is computed using $f$ and is hence of the form $\sum_{j=1}^d x^j A^{j-1}$, where $x^j$ is an integer with absolute value at most $M$. Thus, if we want to write some $\tau'(S)$ as $\frac{1}{D} \sum_{j=1}^d x^j A^{j-1}$, then we can choose $D = \det(B)$. Since $D B^{-1}$ is an integer matrix and all $\delta'_i$ are also integers, we obtain that $x^j$ is an integer for every $j$. Furthermore, $|x^j|$ is at most $n n! |D| M \leq n!(n+1)! M$. As a result, $\tau'(S)$ must be proper for all $S \in \Fam$.

Since $f$ preserves the monotonicity of a solution, the monotone MCFT problem in $d$ dimensions can also be reduced to the monotone MCFT problem in $1$ dimension. Furthermore, the numbers in $\tau'$ are integers with absolute value at most $A^d$. Hence, the bit complexity of these numbers is at most $\log(A^d) = d \log(3 (n+1)!^2 M) = O(d (n \log n + \log M))$. This implies that, even if $d$ depends polynomially on $n$, the reduction can be computed in polynomial time.  
\end{proof}

The conditions of Theorem~\ref{thm:MFCTdimreduce} also apply to the instances in the reduction in Theorem~\ref{thm:MFCTNPhard}, and hence we can conclude that the monotone MCFT problem is also NP-hard in $1$ dimension. However, the complexity of the non-monotone version of the problem remains an open problem. Intuitively, this version of the problem seems harder than the monotone version, but the reduction in Theorem~\ref{thm:MFCTNPhard} heavily uses the fact that the problem is monotone. Nonetheless, we believe that the same reduction should also work for the non-monotone version, but that the challenge lies in arguing the correctness of the reduction. We therefore conjecture that the (non-monotone) MCFT problem is also NP-hard.

\section{Conclusion and Future Work}\label{sec:conclusion}
We have introduced a new class of problems for measuring the visual complexity of point set mappings, where group translations are treated as single operations. We are able to establish the algorithmic complexity of most of the problems in our problem classification, except for two: the (non-monotone) MCFT problem and the (Euclidean) MLFT problem. We believe that the former problem is NP-hard, but that is not so clear for the MLFT problem, where we believe a polynomial-time algorithm may be possible. 
Beyond that, our problem classification can be expanded in several natural ways, leading to new open questions:
\begin{description}
    \item[Stages:] Disjoint transformations have the advantage that they can be animated simultaneously, but they are very restricted. We can make them more powerful by considering several disjoint transformations in sequence, as different \emph{stages}. A new optimization criterion may then be to minimize this number of stages in a transformation.
    \item[Unlabeled point sets:] In some cases, the provided point sets may be unlabeled. Additionally identifying a labeling that optimizes the transformation adds another layer of complexity to the described problems.
    \item[Transformations:] Instead of translations we may also consider other natural transformations like rotations and scaling. As there are multiple ways to scale or rotate one point to another, these new transformations increase the complexity of the problem significantly.
\end{description}

\bibliographystyle{splncs04}
\bibliography{Sofsem-arxiv}

\end{document}